\documentclass{article}
\usepackage[utf8]{inputenc}
\usepackage{proof,color}
\usepackage{toyooka_logic}
\usepackage{times}
\usepackage{lineno}
\usepackage{color}
\usepackage{latexsym}
\newcommand{\imp}{\mathbin{\to}}
\theoremstyle{definition}
\newtheorem{definition}{Definition}[]
\newtheorem{corollary}{Corollary}[]
\newtheorem{remark}{Remark}[]

\title{Semantic Incompleteness of Hilbert System for a Combination of Classical and Intuitionistic Propositional Logic}
\author{Masanobu Toyooka and Katsuhiko Sano\thanks{This work of the first author was partially supported by Grant-in-Aid for JSPS Fellows Grant Number JP22J20341. 
The work of the second author was partially supported by JSPS KAKENHI Grant-in-Aid for Scientific Research (B) Grant Number JP22H00597 and (C) Grant number JP19K12113.}}
\date{\today}

\begin{document}

\maketitle





\abstract{This paper shows Hilbert system $(\mathbf{C+J})^{-}$, given by del Cerro and Herzig (1996) is semantically incomplete. This system is proposed as a proof theory for Kripke semantics for a combination of intuitionistic and classical propositional logic, which is obtained by adding the natural semantic clause of classical implication into intuitionistic Kripke semantics. Although Hilbert system $(\mathbf{C+J})^{-}$ contains intuitionistic modus ponens as a rule, it does not contain classical modus ponens. This paper gives an argument ensuring that the system $(\mathbf{C+J})^{-}$ is semantically incomplete because of the absence of classical modus ponens. Our method is based on the logic of paradox, which is a paraconsistent logic proposed by Priest (1979).}


\section{Introduction}
This paper shows semantic incompleteness of Hilbert system $(\mathbf{C+J})^{-}$, given by del Cerro and Herzig~\cite{Cerro1996}. This system was provided for a combination of intuitionistic and classical propositional logic. This combined logic has two implications: intuitionistic one (denoted by ``$\imp_{\mathtt{i}}$'') and classical one (denoted by ``$\imp_{\mathtt{c}}$''). This logic also has falsum, conjunction, and disjunction as connectives which are common to intuitionistic and classical logic. Since this logic is constructed to be a combination of intuitionistic and classical logic, Hilbert system $(\mathbf{C+J})^{-}$ has to enable us to be a conservative extension of both logics.

The way of constructing the combination is easier to understand from a semantic respect. The semantics for this combination is given in~\cite{Humberstone1979,Cerro1996}, and the basic idea is adding classical implication into intuitionistic Kripke semantics, which means the satisfaction relation of classical implication, denoted by ``$\imp_{\mathtt{c}}$'', is given in a Kripke model as follows:

\[
\begin{array}{lll}
w \models_{M} A \imp_{\mathtt{c}} B & \iff & w \models_{M} A \text{ implies } w \models_{M} B,\\ 
\end{array}
\]

\noindent where $M$ is an intuitionistic Kripke model, $w$ is a possible world in $M$, and $R$ is a preorder equipped in $M$. The other parts of the Kripke semantics is the same as that of intuitionistic propositional logic. 

The syntax $\mathcal{L}$ consists of a countably infinite set $\mathsf{Prop}$ of propositional variables and the following logical connectives: falsum $\bot$, disjunction $\lor$, conjunction $\land$, intuitionistic implication $\imp_{\mathtt{i}}$, and classical implication $\imp_{\mathtt{c}}$. We denote by $\mathcal{L}_{\mathbf{C}}$ (the syntax for the classical logic) and $\mathcal{L}_{\mathbf{J}}$ (the syntax for the intuitionistic logic) the resulting syntax dropping $\imp_{\mathtt{i}}$ and $\imp_{\mathtt{c}}$ from $\mathcal{L}$, respectively. 

The set $\mathsf{Form}$ of all formulas in the syntax is defined inductively as follows: 

\[
 A ::= p \mid \bot \mid A \lor A \mid A \land A \mid A \imp_{\mathtt{i}} A \mid A \imp_{\mathtt{c}} A,  
\]

\noindent where $p \in \mathsf{Prop}$. We denote by $\mathsf{Form}_{\mathbf{C}}$ and $\mathsf{Form}_{\mathbf{J}}$ the set of all classical formulas and the set of all intuitionistic formulas, respectively. We define $\top$ := $\bot \imp_{\mathtt{i}} \bot$, $\neg_{\mathtt{c}} A$ := $A \imp_{\mathtt{c}} \bot$, and $\neg_{\mathtt{i}} A$ := $A \imp_{\mathtt{i}} \bot$. 

Let us move to the semantics for the syntax $\mathcal{L}$. 
\begin{definition}
\label{def:model1}
 A {\em model} is a tuple $M$ = $(W,R,V)$ where 
 \begin{itemize}
     \item $W$ is a non-empty set of possible worlds,
     \item $R$ is a preorder on $W$, i.e., $R$ satisfies reflexivity and transitivity,
     \item $V: \mathsf{Prop} \to \mathcal{P}(W)$ is a valuation function satisfying the following {\em heredity} condition: $w \in V(p)$ and $w R v$ jointly imply $v \in V(p)$ for all possible worlds $w, v \in W$.
 \end{itemize}
\end{definition}  
\begin{definition}
\label{def:model2}
Given a model $M$ = $(W,R,V)$, a possible world $w \in W$ and a formula $A$, the {\em satisfaction relation} $w \models_{M} A$ is inductively defined as follows: 
\[
\begin{array}{lll}
w \models_{M} p & \iff& w \in V(p),\\
w \not\models_{M} \bot & \text{}\\
w \models_{M} A \land B & \iff & w \models_{M} A \text{ and } w \models_{M} B,\\
w \models_{M} A \lor B & \iff & w \models_{M} A \text{ or } w \models_{M} B,\\
w \models_{M} A \imp_{\mathtt{i}} B & \iff & \text{for all } v \in W, (wRv \text{ and } v \models_{M} A \text{ jointly imply } v \models_{M} B),\\
w \models_{M} A \imp_{\mathtt{c}} B & \iff & w \models_{M} A \text{ implies } w \models_{M} B.\\
\end{array}
\]
\end{definition}
\noindent Let $\Gamma$ be a set of formulas and $A$ be a formula.
A formula $A$ is called a {\em semantic consequence} of $\Gamma$, written as $\Gamma \models A$, if, for all models $M$ = $(W,R,V)$ and all worlds $w \in W$, $w \models_{M} C$ for all formulas $C \in \Gamma$  implies $w \models_{M} A$.
A formula $A$ is {\em valid} if $\emptyset \models A$ holds.

\noindent From Definition \ref{def:model2}, the following satisfaction relation for a formula whose main connective is $\neg_{\mathtt{c}}$ or $\neg_{\mathtt{i}}$ is obtained:
\[
\begin{array}{lll}
w \models_{M} \neg_{\mathtt{c}} A & \iff & w \not\models_{M} A,\\
w \models_{M} \neg_{\mathtt{i}} A & \iff & \text{for all } v \in W, (wRv \text{ implies } v \not\models_{M} A). \\
\end{array}
\]

A notion of heredity, which is an important notion in intuitionistic logic, can be defined in this Kripke semantics.

\begin{dfn}[Heredity]
\label{dfn:her}
A formula $A$ satisfies {\em heredity} iff for any model $M$ and $w,v \in W$, $wRv$ and $w \models_{M} A$ jointly imply $v \models_{M} A$.
\end{dfn}

\noindent In pure intuitionistic logic, any formula satisfies heredity. However, if we add classical implication, there will exist a formula which does not satisfy heredity.

\begin{prop}
\label{prop:breher}
A formula $\neg_{\mathtt{c}} p$ does not satisfy heredity.
\end{prop}

\noindent Corresponding to Proposition \ref{prop:breher}, the following proposition also holds.

\begin{prop}
\label{prop:breher2}
Both $\neg_{\mathtt{c}} p \imp_{\mathtt{c}} (q \imp_{\mathtt{i}} \neg_{\mathtt{c}} p)$ and $\neg_{\mathtt{c}} p \imp_{\mathtt{i}} (q \imp_{\mathtt{i}} \neg_{\mathtt{c}} p)$ is {\em invalid}.
\end{prop}

\noindent Proposition \ref{prop:breher2} implies an intuitionistic theorem $A \imp_{\mathtt{i}} (B \imp_{\mathtt{i}} A)$ is no longer a theorem in this combination. An argument about Propositions \ref{prop:breher} and \ref{prop:breher2} was given in~\cite{Toyooka2021a,Toyooka2022b}.

Let us move to a proof theory given in~\cite{Cerro1996}. 
Before giving the detail of their axiomatization, we introduce the notion of {\em persistent formulas} as follows:
\[
E ::= \bot \mid p \mid A \imp_{\mathtt{i}} A \mid E \land E \mid E \lor E,
\]
where $p \in \mathsf{Prop}$ and $A \in \mathsf{Form}$. 
\footnote {Del Cerro and Herzig~\cite{Cerro1996} did not defined $\bot$ as a persistent formula.
However, this is not an essential point, since $\bot$ is equivalent to $p \land \neg_{\mathtt{i}} p$, which is a persistent formula in the sense of~\cite{Cerro1996}. This slight change allows us to state that all formulas of $\mathcal{L}_{\mathbf{J}}$ is persistent.} 

\begin{dfn}
Hilbert system $(\mathbf{C+J})^{-}$ consists of axioms 
$(\mathtt{CL})$, $(\mathtt{CK})$, $(\mathtt{ID})$, $(\mathtt{CMP})$, and $(\mathtt{PER})$ of Table \ref{table:minus} and rules $(\mathtt{MPI})$ and $(\mathtt{RCN})$ of Table \ref{table:minus}. 
Hilbert system $\mathbf{C+J}$ is the extended system of $(\mathbf{C+J})^{-}$ with the rule $(\mathtt{MPC})$ of Table \ref{table:minus}.
\end{dfn}

\begin{table}[htbp]
\centering
\caption{Hilbert Systems $(\mathbf{C+J})^{-}$ and $\mathbf{C+J}$}
\label{table:minus}
\begin{tabular}{|cl|}
    \hline
    \multicolumn{2}{|c|}{ Hilbert System $(\mathbf{C+J})^{-}$ } \\
     (\texttt{CL}) & All instances of classical tautologies \\
     (\texttt{CK})& $(A \imp_{\mathtt{i}} (B \imp_{\mathtt{c}} C)) \imp_{\mathtt{c}} ((A \imp_{\mathtt{i}} B) \imp_{\mathtt{c}} (A \imp_{\mathtt{i}} C))$ \\ (\texttt{ID})& $A \imp_{\mathtt{i}} A$ \\
     (\texttt{CMP}) &  $(A \imp_{\mathtt{i}} B) \imp_{\mathtt{c}} (A \imp_{\mathtt{c}} B)$ \\
     (\texttt{PER}) & $A \imp_{\mathtt{c}} (B \imp_{\mathtt{i}} A)^{\dagger}$ \quad
     $\dagger$: $A$ is {\em persistent}. \\
     (\texttt{MPI}) & From $A$ and $A \imp_{\mathtt{i}} B$ we may infer $B$ \\
    (\texttt{RCN}) & From $A$ we may infer $B \imp_{\mathtt{i}} A$ \\
    \hline
    \multicolumn{2}{|c|}{ Hilbert System $\mathbf{C+J}$ }\\
    &  All the axioms and rules of $(\mathbf{C+J})^{-}$\\ 
    {(\texttt{MPC})}& From $A$ and $A \imp_{\mathtt{c}} B$ we may infer $B$ \\ \hline
\end{tabular}
\end{table}

\noindent An important axiom is (\texttt{PER}). Recall that a formula $\neg_{\mathtt{c}} p \imp_{\mathtt{i}} (q \imp_{\mathtt{i}} \neg_{\mathtt{c}} p)$ is not valid in the Kripke semantics, as is described in Proposition \ref{prop:breher2}. In order for this formula to be underivable, an antecedent formula $A$ of (\texttt{PER}) should be restricted to a persistent formula.

In the next section, we show Hilbert system $(\mathbf{C+J})^{-}$ is semantically incomplete. This semantic incompleteness is the result of the absence of (\texttt{MPC}). The system $\mathbf{C+J}$ may be what del Cerro and Herzig~\cite{Cerro1996} intended to provide, but (\texttt{MPC}) does not exist, which may be an unfortunate typo.
It is reasonable that (\texttt{MPC}) is necessary, because $\mathbf{C+J}$ is based on an idea of an axiomatization of conditional logic, which adds axioms and rules on the conditional on the top of classical tautologies and the rule of classical modus ponens (see, e.g.,~\cite{Chellas1975,chellas_1980,Olivetti2007}).

The following proposition ensures the rule \texttt{(MPI)} can be deleted from Hilbert system $\mathbf{C+J}$.

\begin{prop}
If we drop $\texttt{(MPI)}$ from $\mathbf{C+J}$, 
$\texttt{(MPI)}$ is derivable in the resulting system.
\end{prop}

\begin{proof}
Suppose $A$ and $A \imp_{\mathtt{i}} B$ are theorems in the resulting system. By \texttt{(CMP)} and
\texttt{(MPC)}, $A \imp_{\mathtt{c}} B$ is obtained from $A \imp_{\mathtt{i}} B$ . By applying \texttt{(MPC)} to $A$ and $A \imp_{\mathtt{c}} B$, $B$ is obtained, as is desired.
\end{proof}

\noindent Thus, in order to obtain $\mathbf{C+J}$ from $(\mathbf{C+J})^{-}$, to replace (\texttt{MPI}) with (\texttt{MPC}) is sufficient.

\section{Semantic Incompleteness of Hilbert System $(\mathbf{C + J})^{-}$}
\label{sec:app}
In this section, we show Hilbert system $(\mathbf{C+J})^{-}$ is semantically incomplete. 
We provide a formula $C$ such that $C$ is valid in the semantics described in Definitions \ref{def:model1} and \ref{def:model2} but $C$ is not a theorem of $(\mathbf{C+J})^{-}$.
Our candidate for $C$ is $(p \land (p \imp_{\mathtt{c}} q)) \imp_{\mathtt{i}} q$. The following is easy to establish.

\begin{prop}
\label{fact:ppq}
The formula $(p \land (p \imp_{\mathtt{c}} q)) \imp_{\mathtt{i}} q$ is valid in Kripke semantics in Definitions \ref{def:model1} and \ref{def:model2}.
\end{prop}

\noindent What we need to show is to establish that $(p \land (p \imp_{\mathtt{c}} q)) \imp_{\mathtt{i}} q$ is not a theorem of $(\mathbf{C+J})^{-}$. 
For this purpose, we need to consider a non-standard semantics such that the soundness holds to the original system  but $(p \land (p \imp_{\mathtt{c}} q)) \imp_{\mathtt{i}} q$ is not valid. 

In order to make such a semantics, we utilize three-valued semantics for a paraconsistent logic by Priest (cf.~\cite{Priest2002}), i.e., the logic of paradox~\cite{Priest1979}, which allows the third truth value of ``both true and false'' $\setof{0,1}$ in addition to the values $\setof{0}$ (``false only'') and $\setof{1}$ (``true only'').  

\begin{definition}
\label{def:non}
A valuation $v$ is a mapping from $\mathsf{Prop}$ to $\{ \{0 \}, \{ 0,1 \}, \{ 1 \} \}$. A valuation $v$ is uniquely extended to a function $\overline{v}$ from the set $\mathsf{Form}$ of all formulas to $\{ \{0 \}, \{ 0,1 \}, \{ 1 \} \}$ as follows: 
\[
\begin{array}{rll}
1 \in \overline{v}(\bot) &  & \text{Never}, \\ 0 \in \overline{v}(\bot) &  & \text{Always},\\ 1 \in \overline{v} (A \land B) & \iff& 1 \in \overline{v}(A) \text{ and } 1 \in \overline{v}(B), \\
0 \in \overline{v} (A \land B) & \iff& 0 \in \overline{v}(A) \text{ or } 0 \in \overline{v}(B), \\
1 \in \overline{v} (A \lor B) & \iff& 1 \in \overline{v}(A) \text{ or } 1 \in \overline{v}(B), \\
0 \in \overline{v} (A \lor B) & \iff& 0 \in \overline{v}(A) \text{ and } 0 \in \overline{v}(B), \\
1 \in \overline{v} (A \imp_{\mathtt{c}} B) & \iff& 0 \in \overline{v}(A) \text{ or } 1 \in \overline{v}(B), \\
0 \in \overline{v} (A \imp_{\mathtt{c}} B) & \iff& 1 \in \overline{v}(A) \text{ and } 0 \in \overline{v}(B), \\
1 \in \overline{v} (A \imp_{\mathtt{i}} B) & \iff& 1 \notin \overline{v}(A) \text{ or } 1 \in \overline{v}(B), \\
0 \in \overline{v} (A \imp_{\mathtt{i}} B) & \iff& 1 \in \overline{v}(A) \text{ and } 0 \in \overline{v}(B). \\
\end{array}
\]
A consequence relation $\Sigma \models_{3} A$ is defined as: if $1 \in \overline{v}(B)$ holds for all $B \in \Sigma$ then $1 \in \overline{v}(A)$. We say that a formula $A$ is {\em $3$-valid} if $\models_{3} A$ holds. 
\end{definition}

\begin{prop}
\label{prop:lem}
For every valuation $v: \mathsf{Prop} \to \setof{\setof{1}, \setof{0,1}, \setof{0}}$ and every $A \in \mathsf{Form}$, $1 \in \overline{v}(A)$ or $0 \in \overline{v}(A)$. 
\end{prop}

\begin{proof}
Fix any valuation $v: \mathsf{Prop} \to \setof{\setof{1}, \setof{0,1}, \setof{0}}$. 
By induction on $A$, we can obtain the desired statement. 
\end{proof}
    
\begin{remark}
Let us denote $\{ 1 \}, \{ 0,1 \}, \{ 0 \}$ by $\mathbf{t}, \mathbf{b} ,\mathbf{f}$, respectively. Then the semantics above provides the three valued truth table, described in Table \ref{table:truth}, where the values $\mathbf{t}$ and $\mathbf{b}$ are defined as ``designated values''.

\begin{table}[ht]
\caption{Three Valued Truth Table}
\label{table:truth}
\centering
    \begin{tabular}{c|c|c|c} 
         $\land$ & $\mathbf{t}$ & $\mathbf{b}$ & $\mathbf{f}$ \\ \hline
         $\mathbf{t}$ & $\mathbf{t}$ & $\mathbf{b}$ & $\mathbf{f}$ \\ \hline
         $\mathbf{b}$ & $\mathbf{b}$ & $\mathbf{b}$ & $\mathbf{f}$ \\ \hline
         $\mathbf{f}$ & $\mathbf{f}$ & $\mathbf{f}$ & $\mathbf{f}$ \\ \hline 
    \end{tabular}
    \quad
    \begin{tabular}{c|c|c|c} 
         $\lor$ & $\mathbf{t}$ & $\mathbf{b}$ & $\mathbf{f}$ \\ \hline
         $\mathbf{t}$ & $\mathbf{t}$ & $\mathbf{t}$ & $\mathbf{t}$ \\ \hline
         $\mathbf{b}$ & $\mathbf{t}$ & $\mathbf{b}$ & $\mathbf{b}$ \\ \hline
         $\mathbf{f}$ & $\mathbf{t}$ & $\mathbf{b}$ & $\mathbf{f}$ \\ \hline 
    \end{tabular}
    \quad
    \begin{tabular}{c|c|c|c} 
         $\imp_{\mathtt{c}}$ & $\mathbf{t}$ & $\mathbf{b}$ & $\mathbf{f}$ \\ \hline
         $\mathbf{t}$ & $\mathbf{t}$ & $\mathbf{b}$ & $\mathbf{f}$ \\ \hline
         $\mathbf{b}$ & $\mathbf{t}$ & $\mathbf{b}$ & $\mathbf{b}$ \\ \hline
         $\mathbf{f}$ & $\mathbf{t}$ & $\mathbf{t}$ & $\mathbf{t}$ \\ \hline 
    \end{tabular}
    \quad
     \begin{tabular}{c|c|c|c} 
         $\imp_{\mathtt{i}}$ & $\mathbf{t}$ & $\mathbf{b}$ & $\mathbf{f}$ \\ \hline
         $\mathbf{t}$ & $\mathbf{t}$ & $\mathbf{b}$ & $\mathbf{f}$ \\ \hline
         $\mathbf{b}$ & $\mathbf{t}$ & $\mathbf{b}$ & $\mathbf{f}$ \\ \hline
         $\mathbf{f}$ & $\mathbf{t}$ & $\mathbf{t}$ & $\mathbf{t}$ \\ \hline 
    \end{tabular}
    \quad
    \begin{tabular}{c|c}
          & $\bot$ \\ \hline
         $\mathbf{t}$ & $\mathbf{f}$ \\ \hline
         $\mathbf{b}$ & $\mathbf{f}$ \\ \hline
         $\mathbf{f}$ & $\mathbf{f}$ \\ \hline
    \end{tabular}
\end{table}

\noindent By recalling $\neg_{\mathtt{c}} A$ := $A \imp_{\mathtt{c}} \bot$  and $\neg_{\mathtt{i}} A$ := $A \imp_{\mathtt{i}} \bot$ respectively, we can also obtain the following satisfaction relation for negations:
\[
\begin{array}{lll}
1 \in \overline{v}(\neg_{\mathtt{c}} A) & \iff & 0 \in \overline{v}(A), \\
0 \in \overline{v}(\neg_{\mathtt{c}} A) & \iff & 1 \in \overline{v}(A), \\
1 \in \overline{v}(\neg_{\mathtt{i}} A) & \iff & 1 \notin \overline{v}(A), \\
0 \in \overline{v}(\neg_{\mathtt{i}} A) & \iff & 1 \in \overline{v}(A).
\end{array}
\]
Therefore, $\neg_{\mathtt{c}}$ is De Morgan negation (cf.~\cite{Priest2002}). We can also get the truth table for $\neg_{\mathtt{c}}$ and $\neg_{\mathtt{i}}$, described in Table \ref{tab:neg}.

\begin{table}[htbp]
\caption{Truth Table for Negations}
\label{tab:neg}
    \centering
     \begin{tabular}{c|c}
         $A$ & $\neg_{\mathtt{c}}A$  \\ \hline
         $\mathbf{t}$ & $\mathbf{f}$ \\ \hline
         $\mathbf{b}$ & $\mathbf{b}$ \\ \hline
         $\mathbf{f}$ & $\mathbf{t}$ \\ \hline
    \end{tabular}
    \quad
    \begin{tabular}{c|c}
         $A$ &$\neg_{\mathtt{i}}A$  \\ \hline
         $\mathbf{t}$ & $\mathbf{f}$ \\ \hline
         $\mathbf{b}$ & $\mathbf{f}$ \\ \hline
         $\mathbf{f}$ & $\mathbf{t}$ \\ \hline
    \end{tabular}
\end{table}

 The set $\setof{\land,\lor, \neg_{\mathtt{c}}}$ of logical connectives is exactly the same set of primitive logical connectives as the propositional part of the logic of paradox~\cite{Priest1979,Priest2002}. 
It is remarked, however, that $A \imp_{\mathtt{c}} B$ is defined as $\neg_{\mathtt{c}} A \lor B$ but $\bot$ cannot be defined in terms of $\setof{\land,\lor, \neg_{\mathtt{c}}}$ (if a formula $B$ defined $\bot$, $B$ would return the value $\mathbf{b}$ for a valuation sending all propositional variables to $\mathbf{b}$, a contradiction). 
Therefore, in terms of three-valued semantics above, $\setof{\land,\lor, \imp_{\mathtt{c}}, \bot}$, i.e., the syntax of $\mathcal{L}_\mathbf{C}$ is stronger than $\setof{\land,\lor, \neg_{\mathtt{c}}}$, i.e, the syntax for the logic of paradox.

The truth table for $\imp_{\mathtt{i}}$ is the same as that of an implication introduced as ``internal implication'' in~\cite{Avron1991a}. This connective are studied also in~\cite{Arieli1996,Arieli1998,Carnielli2000,Carnielli2007,Omori2012,Omori2015}.
\end{remark}

\begin{lem}
\label{lem:mpc_x}
$p, p \imp_{\mathtt{c}} q \not\models_{3} q$
\end{lem}

\begin{proof}
Take a valuation $v$ such that $v(p)$ = $\setof{0,1}$ and $v(q)$ = $\setof{0}$. 
Then, $\overline{v}(p \imp_{\mathtt{c}} q)$ = $\setof{0,1}$. But, $1 \notin v(q)$. 
Therefore, $p, p \imp_{\mathtt{c}} q \not\models_{3} q$.
\end{proof}

While the logic of paradox~\cite{Priest1979,Priest2002} has a different consequence relation from those of classical logic (as shown in Lemma \ref{lem:mpc_x}), 
it is well-known that the logic of paradox has the same theorems as those of classical logic (see, e.g., \cite[p.310]{Priest2002}). 
We may also extend this fact to $\setof{\land,\lor, \imp_{\mathtt{c}}, \bot}$, i.e., the syntax of $\mathcal{L}_\mathbf{C}$ as follows.

\begin{prop}
\label{prop:cllp}
Let $A \in \mathsf{Form}_{\mathbf{C}}$. Then, $A$ is a tautology in classical logic iff $A$ is 3-valid.
\end{prop}

\begin{proof}
The proof from right to left is trivial, since $v: \mathsf{Prop} \to \setof{0,1}$ is regarded as a three valued valuation by regarding $0$ and $1$ with $\setof{0}$ and $\setof{1}$, respectively. 
Conversely, assume that $A$ is a tautology and fix any valuation $v: \mathsf{Prop} \to \{ \{0 \}, \{ 0,1 \}, \{ 1 \} \}$. 
Our goal is to show that $1 \in \overline{v}(A)$. 
Define a valuation $v_{1}$ from $v$ by changing all outputs $\{ 0, 1 \}$ of $v$ to $\{ 1 \}$. 
We regard $v_{1}$ as a two-valued valuation function by regarding $\setof{0}$ and $\setof{1}$ with $0$ and $1$, respectively. 
It is easy to see that $v_{1}(p) \subseteq v(p)$ for all $p \in \mathsf{Prop}$. 
We also have $\overline{v}_{1}(B) \subseteq \overline{v}(B)$ for all $B \in \mathsf{Form_{\mathbf{C}}}$ by Definition \ref{def:non} (recall that $C \imp_{\mathtt{c}} D$ is equivalent with $\neg_{\mathtt{c}} C \lor D$). 
Since $A$ is a classical tautology, then $1 \in \overline{v}_{1} (A)$. 
By $\overline{v}_{1}(A) \subseteq \overline{v}(A)$, we conclude that $1 \in \overline{v}(A)$.
\end{proof}

\begin{lem}
 \label{lem:sub}
 Let $A \in \mathsf{Form}_{\mathbf{C}}$. If $\models_{3} A$ then $\models_{3} \sigma(A)$ 
 for all uniform substitutions $\sigma: \mathsf{Prop} \to \mathsf{Form}$. 
\end{lem}
\begin{proof}
Assume that $\models_{3} A$.
Fix any valuation ${v}:\mathsf{Prop} \to \{ \{1 \}, \{ 0,1 \}, \{ 0 \} \}$ and any uniform substitution $\sigma: \mathsf{Prop} \to \mathsf{Form}$. 
The goal is to show $1 \in \overline{v}(\sigma(A))$. 
Define ${v'}:\mathsf{Prop} \to \{ \{1 \}, \{ 0,1 \}, \{ 0 \} \}$ as follows:
\[
\begin{array}{lll}
1 \in {v'} (p) & \iff& 1 \in \overline{v}(\sigma(p)), \\
0 \in {v'} (p) & \iff& 0 \in \overline{v}(\sigma(p)), \\
\end{array}
\]  
for all $p \in \mathsf{Prop}$. 
By assumption, we have $1 \in \overline{v'}(A)$. 
By induction on a formula $B$, we can establish:
\[
\begin{array}{lll}
1 \in {\overline{v'}} (B) & \iff& 1 \in \overline{v}(\sigma(B)), \\
0 \in {\overline{v'}} (B) & \iff& 0 \in \overline{v}(\sigma(B)). \\
\end{array}
\]  
Here, we only deal with the case where $B$ is of the form $C \imp_{\mathtt{i}} D$: 
\[
\begin{array}{lll}
1 \in \overline{v'} (C \imp_{\mathtt{i}} D) & \iff & 1 \notin \overline{v'}(C) \text{ or } 1 \in \overline{v'}(D), \\
 & \iff & 1 \notin \overline{v}(\sigma(C)) \text{ or }1 \in \overline{v}(\sigma(D)), \text{ by induction hypothesis,}\\
 & \iff & 1 \in \overline{v} (\sigma(C) \imp_{\mathtt{i}} \sigma(D)), \\ 
 & \iff & 1 \in \overline{v} (\sigma(C \imp_{\mathtt{i}} D)). \\ 
\end{array}
\]  
\[
\begin{array}{lll}
0 \in \overline{v'} (C \imp_{\mathtt{i}} D) & \iff & 1 \in \overline{v'}(C) \text{ and } 0 \in \overline{v'}(D), \\
 & \iff &  1 \in \overline{v}(\sigma(C)) \text{ and } 0 \in \overline{v}(\sigma(D)), \text{ by induction hypothesis,}\\
 & \iff & 0 \in \overline{v} (\sigma(C) \imp_{\mathtt{i}} \sigma(D)), \\ 
 & \iff & 0 \in \overline{v} (\sigma(C \imp_{\mathtt{i}} D)). \\ 
\end{array}
\]
Since we have $1 \in \overline{v'}(A)$, we conclude that $1 \in \overline{v}(\sigma(A))$, as required. 
\end{proof}
\begin{lem}
\label{lem:lps}
Let $A \in \mathsf{Form}$. If $A$ is a theorem of $\mathbf{(C + J)^{-}}$, then $A$ is 3-valid.
\end{lem}

\begin{proof}
It suffices to show each axiom of $\mathbf{(C + J)}^{-}$ is 3-valid and each rule of the system preserves 3-validity. 
\begin{itemize}
    \item $\texttt{(CL)}$ Let $A$ be an instance of a classical tautology, i.e., $A$ is of the form $\sigma(A')$ where $A' \in \mathsf{Form}_{\mathbf{C}}$ is a classical tautology and $\sigma: \mathsf{Prop} \to \mathsf{Form}$ is a uniform substitution, where $\mathsf{Form}$ is the set of all formulas of the syntax for $(\mathbf{C+J})^{-}$. 
    Our goal is to show $\models_{3} \sigma(A')$. 
    By Lemma \ref{lem:sub}, it suffices to show $\models_{3} A'$. Since $A' \in \mathsf{Form}_{\mathbf{C}}$, Proposition \ref{prop:cllp} tells us that we need to establish that $A'$ is a classical tautology. But, this is our assumption. 
    \item $\texttt{(ID)}$ This is trivial, since $1 \in \overline{v}(A \to B)$ iff $1 \in \overline{v}(A)$ implies $1 \in \overline{v}(B)$. 
    \item $\texttt{(CK)}$ We show $\models_{3} (A \imp_{\mathtt{i}}(B \imp_{\mathtt{c}} C)) \imp_{\mathtt{c}} ((A \imp_{\mathtt{i}} B) \imp_{\mathtt{c}} (A \imp_{\mathtt{i}} C))$. Fix any valuation $v: \mathsf{Prop} \to \setof{\setof{1}, \setof{0,1}, \setof{0}}$. Our goal is to show $1 \in \overline{v} ((A \imp_{\mathtt{i}}(B \imp_{\mathtt{c}} C)) \imp_{\mathtt{c}} ((A \imp_{\mathtt{i}} B) \imp_{\mathtt{c}} (A \imp_{\mathtt{i}} C)))$. It suffices to show that $0 \notin \overline{v} (A \imp_{\mathtt{i}}(B \imp_{\mathtt{c}} C))$ implies $1 \in \overline{v}((A \imp_{\mathtt{i}} B) \imp_{\mathtt{c}} (A \imp_{\mathtt{i}} C))$. Suppose $0 \notin \overline{v} (A \imp_{\mathtt{i}}(B \imp_{\mathtt{c}} C))$. This implies $1 \notin \overline{v}(A)$ or $0 \notin \overline{v}(B \imp_{\mathtt{c}} C)$. 
    For each case, we establish $1 \in \overline{v}((A \imp_{\mathtt{i}} B) \imp_{\mathtt{c}} (A \imp_{\mathtt{c}} B))$, i.e., 
    $0 \in \overline{v}(A \imp_{\mathtt{i}} B)$ or 
    $1 \in \overline{v}(A \imp_{\mathtt{c}} C)$. 
    If $1 \notin \overline{v}(A)$ holds, then $1 \in \overline{v}(A \imp_{\mathtt{i}} C)$ holds hence $1 \in \overline {v}((A \imp_{\mathtt{i}} B) \imp_{\mathtt{c}} (A \imp_{\mathtt{i}} C))$. If $0 \notin \overline{v}(B \imp_{\mathtt{c}} C)$ holds, then $1 \notin \overline{v}(B)$ or $0 \notin \overline{v}(C)$ holds. 
    Without loss of generality, we can also assume $1 \in \overline{v}(A)$. 
    If $1 \notin \overline{v}(B)$ holds, we derive from Proposition \ref{prop:lem} that $0 \in \overline{v}(B)$. With $1 \in \overline{v}(A)$, it implies $0 \in \overline{v}(A \imp_{\mathtt{i}} B)$. Therefore, we can obtain $1 \in \overline {v}((A \imp_{\mathtt{i}} B) \imp_{\mathtt{c}} (A \imp_{\mathtt{i}} C))$. If $0 \notin \overline{v}(C)$ holds, then we deduce from Proposition \ref{prop:lem} that $1 \in \overline{v}(C)$. By this, we can obtain $1 \in \overline{v}(A \imp_{\mathtt{i}} C)$. Therefore, $1 \in \overline {v}((A \imp_{\mathtt{i}} B) \imp_{\mathtt{c}} (A \imp_{\mathtt{i}} C))$ holds. This finishes our argument by cases. 
    \item $\texttt{(CMP)}$ We show $\models_{3} (A \imp_{\mathtt{i}} B) \imp_{\mathtt{c}} (A \imp_{\mathtt{c}} B)$. Fix any valuation $v: \mathsf{Prop} \to \setof{\setof{1}, \setof{0,1}, \setof{0}}$. Our goal is to show $1 \in \overline{v}((A \imp_{\mathtt{i}} B) \imp_{\mathtt{c}} (A \imp_{\mathtt{c}} B))$. It suffices to show that $0 \notin \overline{v}(A \imp_{\mathtt{i}} B)$ implies $1 \in \overline{v}(A \imp_{\mathtt{c}} B)$. Suppose $0 \notin \overline{v}(A \imp_{\mathtt{i}} B)$. 
    This implies $1 \notin \overline{v}(A)$ or $0 \notin \overline{v}(B)$. 
    If $1 \notin \overline{v}(A)$ holds, then, by Proposition \ref{prop:lem}, $0 \in \overline{v}(A)$ holds hence $1 \in \overline{v}(A \imp_{\mathtt{c}} B)$. If $0 \notin \overline{v}(B)$, then we deduce from Proposition \ref{prop:lem} that $1 \in \overline{v}(B)$ hence $1 \in \overline{v}(A \imp_{\mathtt{c}} B)$. 
    For both cases, we have established $1 \in \overline{v}(A \imp_{\mathtt{c}} B)$.
    \item $\texttt{(PER)}$ For this validity, we do not have to impose any restriction on $A$ here. We show $\models_{3} A \imp_{\mathtt{c}} (B \imp_{\mathtt{i}} A)$. Fix any valuation $v: \mathsf{Prop} \to \setof{\setof{1}, \setof{0,1}, \setof{0}}$. Our goal is to show $1 \in \overline{v}(A \imp_{\mathtt{c}}(B \imp_{\mathtt{i}} A))$. It suffices to show that $0 \notin \overline{v} (A)$ implies $1 \in \overline{v} (B \imp_{\mathtt{i}} A)$. Suppose  $0 \notin \overline{v}(A)$. By Proposition \ref{prop:lem}, we get $1 \in \overline{v}(A)$ hence $1 \in \overline{v}(B \imp_{\mathtt{i}} A)$, as required. 
    \item $\texttt{(MPI)}$ We show that $\models_{3} A$ and $\models_{3} A \imp_{\mathtt{i}} B$ imply $\models_{3} B$. Suppose $\models_{3} A$ and $\models_{3} A \imp_{\mathtt{i}} B$. Our goal is to show $\models_{3} B$. Fix any valuation $v: \mathsf{Prop} \to \setof{\setof{1}, \setof{0,1}, \setof{0}}$. We show $1 \in \overline{v}(B)$. By the supposition, we have $1 \in \overline{v}(A)$ and $1 \in \overline{v}(A \imp_{\mathtt{i}}B)$. Since $1 \in \overline{v}(A)$ and $1 \notin \overline{v}(A)$ are not compatible, we can deduce from $1 \in \overline{v}(A \imp_{\mathtt{i}}B)$ that $1 \in \overline{v}(B)$, as desired.
    \item $\texttt{(RCN)}$ We show that $\models_{3} A$ implies $\models_{3} B \imp_{\mathtt{i}} A$. Suppose $\models_{3} A$. Our goal is to show $\models_{3} B \imp_{\mathtt{i}} A$. Fix any valuation $v: \mathsf{Prop} \to \setof{\setof{1}, \setof{0,1}, \setof{0}}$. We show $1 \in \overline{v}(B \imp_{\mathtt{i}} A)$. By the supposition, $1 \in \overline{v}(A)$ holds. Therefore, we can obtain $1 \in \overline{v}(B \imp_{\mathtt{i}} A)$ straightforwardly by Definition \ref{def:non}.\qedhere
    \end{itemize}

\end{proof}

\begin{lem}
 \label{lem:mod}
$\models_{3} (p \land (p \imp_{\mathtt{c}} q)) \imp_{\mathtt{i}} q$ iff $p, p \imp_{\mathtt{c}} q \models_{3} q$.
\end{lem}

\begin{proof}
 This follows from the equivalence: $1 \in \overline{v}(A \imp_{\mathtt{i}}B)$ iff $1 \in \overline{v}(A)$ implies $1 \in \overline{v}(B)$. 
\end{proof}

\begin{thm}
\label{thm:noc}
The formula $(p \land (p \imp_{\mathtt{c}} q)) \imp_{\mathtt{i}} q$ is not a theorem in $\mathbf{(C + J)}^{-}$.
\end{thm}

\begin{proof}
 Suppose $(p \land (p \imp_{\mathtt{c}} q)) \imp_{\mathtt{i}} q$ is a theorem in $\mathbf{(C + J)}^{-}$. By Lemma \ref{lem:lps}, $(p \land (p \imp_{\mathtt{c}} q)) \imp_{\mathtt{i}} q$ is 3-valid. By Lemma \ref{lem:mod}, $p, p \imp_{\mathtt{c}} q \models_{3} q$ should hold. This is a contradiction with Lemma \ref{lem:mpc_x}. 
\end{proof}

\begin{corollary}
Hilbert system $\mathbf{(C + J)}^{-}$ is not semantically complete, i.e., 
there exists a formula $C$ such that 
$C$ is not a theorem of $\mathbf{(C + J)}^{-}$ but $C$ is valid in Kripke semantics in Definition \ref{def:model2}. 
\end{corollary}

\begin{proof}
 By Proposition \ref{fact:ppq} and Theorem \ref{thm:noc}.
\end{proof}

\noindent The argument described above implies, in order to obtain the completeness theorem, the rule \texttt{(MPC)} is necessary. If \texttt{(MPC)} is added, Theorem \ref{thm:noc} will no longer hold. This is because Lemma \ref{lem:lps} does not hold for $\mathbf{C+J}$, since (\texttt{MPC}) does not preserve 3-validity, which is a well-known feature of the logic of paradox.

\bibliographystyle{plain}

\begin{thebibliography}{10}

\bibitem{Arieli1996}
Ofer Arieli and Arnon Avron.
\newblock Reasoning with logical bilattices.
\newblock {\em Journal of Logic, Language and Information}, 5:25--63, March
  1996.

\bibitem{Arieli1998}
Ofer Arieli and Arnon Avron.
\newblock The value of the four values.
\newblock {\em Artificial Intelligence}, 102(1):97--141, June 1998.

\bibitem{Avron1991a}
Arnon Avron.
\newblock Natural 3-valued logics: characterization and proof theory.
\newblock {\em The Journal of Symbolic Logic}, 56(1):276--294, March 1991.

\bibitem{Carnielli2007}
Walter~A Carnielli, Marcelo E.~Coniglio, and Jo{\~a}o Marcos.
\newblock Logics of formal inconsistency.
\newblock In Dov~M Gabbay and Franz Guenthner, editors, {\em Handbook of
  Philosophical Logic}, volume~14 of {\em Handbook of Philosophical Logic
  (HALO)}, pages 1--93. Springer, 2007.

\bibitem{Carnielli2000}
Walter~A Carnielli, Jo{\~a}o Marcos, and Sandra de~Amo.
\newblock Formal inconsistency and evolutionary database.
\newblock {\em Logic and Logical Philosophy}, 8:115--152, 2000.

\bibitem{Chellas1975}
Brian~F. Chellas.
\newblock Basic conditional logic.
\newblock {\em Journal of Philosophical Logic}, 4:133--153, May 1975.

\bibitem{chellas_1980}
Brian~F. Chellas.
\newblock Conditional logic.
\newblock In {\em Modal Logic: An Introduction}, page 268–276. Cambridge
  University Press, 1980.

\bibitem{Cerro1996}
Luis~Fari{\~n}as del Cerro and Andreas Herzig.
\newblock Combining classical and intuitionistic logic or: Intuitionistic
  implication as a conditional.
\newblock In Franz Badder and Klaus~U Schulz, editors, {\em Frontiers of
  Combining Systems: FroCoS 1996}, pages 93--102. Springer, March 1996.

\bibitem{Humberstone1979}
Lloyd Humberstone.
\newblock Interval semantics for tense logic: some remarks.
\newblock {\em Journal of Philosophical Logic}, 8:171--196, 1979.

\bibitem{Olivetti2007}
Nicola Olivetti, Gian~Luca Pozzato, and Camilla~B Schwind.
\newblock A sequent calculus and a theorem prover for standard conditional
  logics.
\newblock {\em ACM Transactions on Computational Logic (TOCL)}, 8(4):22--es,
  August 2007.

\bibitem{Omori2015}
Hitoshi Omori and Katsuhiko Sano.
\newblock Generalizing functional completeness in {B}elnap-{D}unn logic.
\newblock {\em Studia Logica}, 103:883--917, February 2015.

\bibitem{Omori2012}
Hitoshi Omori and Toshiharu Waragai.
\newblock Some observations on the systems {LFI}1 and {LFI}1*.
\newblock In {\em 2011 22nd International Workshop on Detabase and Expert
  Systems Applications}, pages 320--324. IEEE Computer Society, 2012.

\bibitem{Priest1979}
Graham Priest.
\newblock The logic of paradox.
\newblock {\em Journal of Philosophical Logic}, 8(1):219--241, 1979.

\bibitem{Priest2002}
Graham Priest.
\newblock Paraconsistent logic.
\newblock In Dov~M. Gabbay and F.~Guenthner, editors, {\em Handbook of
  Philosophical Logic}, pages 287--393. Springer Netherlands, Dordrecht, 2002.

\bibitem{Toyooka2021a}
Masanobu Toyooka and Katsuhiko Sano.
\newblock Analytic multi-succedent sequent calculus for combining
  intuitionistic and classical propositional logic.
\newblock In Sujata Ghosh and R~Ramanujam, editors, {\em ICLA 2021 Proceedings:
  9th Indian Conference on Logic and its Applications}, pages 128--133. March
  2021.

\bibitem{Toyooka2022b}
Masanobu Toyooka and Katsuhiko Sano.
\newblock Combining first-order classical and intuitionistic logic.
\newblock In Andrzej Indrzejczak and Micha\l Zawidzki, editors, {\em Proceeding
  of the 10th International Conference on Non-Classical Logics. Theory and
  Applications}, volume 358 of {\em Electronic Proceedings in Theoretical
  Computer Science (EPTCS)}, pages 25--40. April 2022.

\end{thebibliography}

\end{document}